\RequirePackage{fix-cm}
\documentclass[twocolumn]{svjour3}         

\smartqed  

\usepackage{graphicx,amsmath,url, amsfonts,bm}      % include this line if your document contains figures
\usepackage{amssymb}
\usepackage{color}

\usepackage{algorithm}
\usepackage{algpseudocode}

\usepackage[caption=false]{subfig}%

\usepackage[sort&compress]{natbib}

\newcommand{\sgn}{\operatorname{sgn}}

\begin{document}

\title{Multi-Observer Output Feedback Stabilization of a Class of Uncertain Nonminimum-Phase Systems}

\author{Roberto Santos \and Kurios Iuri Pinheiro de Melo Queiroz* \and Samaherni Morais Dias \and Tiago Roux Oliveira} 

\institute{Kurios Iuri Pinheiro de Melo Queiroz, Samaherni Morais Dias \at
Federal University of Rio Grande do Norte \\
Natal-RN, Brazil\\
Tiago Roux Oliveira and Roberto Santos  \\
State University of Rio de Janeiro \\
Rio de Janeiro-RJ, Brazil \\
*Corresponding author\\
\email{kurios.queiroz@ufrn.br}  \\
}

\date{Received: date / Accepted: date}

\maketitle

\begin{abstract} 
This paper addresses the challenging problem of output feedback stabilization for nonlinear nonminimum phase (NMP) systems in the presence of parametric uncertainties and external disturbances. The proposed framework integrates three distinct observers: a reduced-order observer for reconstructing the unmeasured states of the internal (zero) dynamics, a high-gain observer for estimating output derivatives, and an observer for estimating the aggregated effect of parametric uncertainties and disturbances. Leveraging these estimates, a sliding mode control law is synthesized to ensure global asymptotic stability of the entire system using only output measurements. The control design requires only partial model knowledge, significantly relaxing the restrictive assumptions common in existing literature. Numerical simulations illustrate the effectiveness of the proposed output-feedback strategy and corroborate the theoretical developments.

\end{abstract}

%===============================================================================
%===============================================================================
%===============================================================================
\section{Introduction}

The challenge of achieving output tracking and stabilization for nonminimum phase (NMP) systems remains a significant topic in control theory, as evidenced by extensive literature, including the historical review in \citep{I:2013} and \citep{fase}. For nonlinear systems, the NMP property is formally defined through the concept of zero dynamics, which generalizes the notion of unstable transmission zeros found in linear system transfer functions \citep{I:2013}. The core difficulty in controlling such systems lies in designing a controller that meets the tracking objectives while simultaneously ensuring the stabilization of the zero dynamics, thereby preventing the divergence of internal states.

%A variety of methodologies have been explored to address this challenge. Proposed solutions include redefining the system output to induce stable zero dynamics \citep{GH:1993, MT:2005}, often leveraging the paradigm of parallel feedforward compensation \citep{KBS:1994}. Other approaches involve the system center method for generating reference trajectories that stabilize the internal dynamics \citep{SS:2001, BSES:2008}, imposing gain constraints for dissipative systems \citep{PIM:2009}, employing high-gain observers for estimating output derivatives \citep{NK:2011}, and utilizing virtual reference feedback tuning \citep{SCB:2018}.

Existing approaches to nonlinear nonminimum-phase control may be broadly classified into three categories. The first consists of output redefinition and feedforward compensation methods \citep{KBS:1994}, which modify the regulated output to obtain stable zero dynamics \citep{GH:1993, MT:2005}. A second class relies on reference-generation techniques, such as system-center methods \citep{PIM:2009}, that produce trajectories compatible with the unstable internal dynamics  \citep{SS:2001, BSES:2008}. A third class employs observer-based feedback designs \citep{NK:2011}, often requiring either partial state measurements or accurate knowledge of the plant dynamics \citep{SCB:2018}.

Despite these significant advances, output-feedback stabilization of uncertain nonminimum-phase systems remains challenging because the controller must simultaneously reconstruct the internal dynamics, estimate uncertain external dynamics, and stabilize both subsystems using only output measurements \citep{RSH:2017, YL:2017, RCF:2017, CK:2017, YDLZ:2018}. Conversely, numerous output-feedback control solutions exist for minimum-phase systems, as seen in \citep{of1, of2, of3, of4, of5}.%On the other hand, there are many solutions for minimum systems using output-feedback control such as \citep{of1,of2,of3,of4,of5}.

%A common limitation of these strategies is their reliance on restrictive assumptions, such as the availability of full or partial state feedback, or a complete/partial knowledge of the system model. Consequently, to handle broader uncertainty, a frequent trade-off is to impose a minimum phase condition \citep{RSH:2017, YL:2017, RCF:2017, CK:2017, YDLZ:2018}. To the authors' knowledge, the problem of overcoming the NMP restriction under significant uncertainty remains largely open. This paper demonstrates that output feedback regulation of NMP systems is achievable with only partial model knowledge through the synergistic use of multiple observers, which constitutes our primary contribution. The presence of parametric uncertainties alongside the NMP property considerably intensifies the difficulty of controller synthesis.

%To this end, we propose a novel adaptive architecture (Fig. \ref{fig: diagrama de blocos}) comprising three distinct observers:

Rather than addressing these challenges through a single observer or assuming additional model information, this paper adopts a modular estimation strategy in which each source of uncertainty is handled by a dedicated observer (Fig. 1).

%\begin{enumerate}
%\item A reduced-order observer dedicated to reconstructing the unmeasured states of the internal dynamics.
%\item An unknown-input observer designed to estimate parametric uncertainties and external disturbances.
%\item A high-gain observer that provides estimates of the output derivatives, which are necessary for implementing the aforementioned observers.
%\end{enumerate}

%Leveraging this multi-observer scheme, we synthesize a sliding mode control law, a well-established robust technique with diverse applications \citep{Balaji02072024,Zhao03072023,Gupta09042024,Gouda01082024}, that guarantees global asymptotic stability for the NMP plant using only output measurements. The choice of sliding modes is motivated by their inherent robustness, which is exploited to counteract parametric uncertainties within the external dynamics, a part of the system often presumed known in existing literature. The efficacy of the proposed control strategy is validated through numerical simulations on an academic example.

The main contributions of this work are summarized as follows:
\begin{itemize}
\item an output-feedback stabilization framework for a class of uncertain nonlinear nonminimum-phase systems;

\item a modular observer architecture combining reduced-order state reconstruction, high-gain output estimation and uncertainty estimation;

\item a Lyapunov-based stability analysis establishing global asymptotic stabilization using only output measurements under mild assumptions;

\item numerical simulations illustrating the effectiveness of the proposed methodology.
\end{itemize}

\begin{figure}[ht]
\centering
%\vspace{-1.5cm}
%\hspace{-1cm}
\includegraphics[width=\columnwidth]{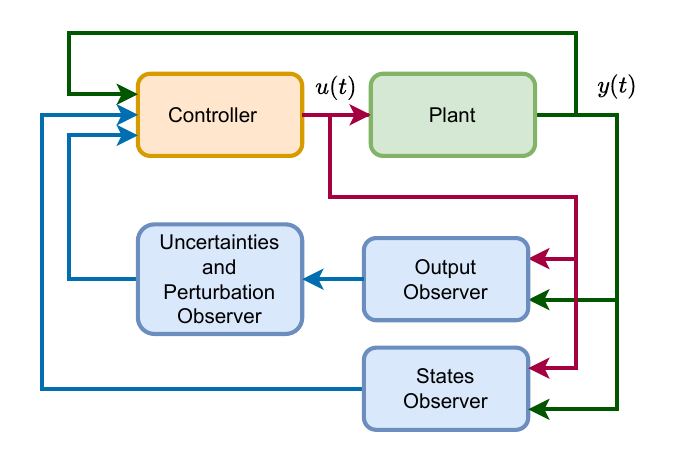}
\caption{Block diagram representing the output feedback with multiple observer scheme.}
\label{fig: diagrama de blocos}
\end{figure}

\subsection{Notations and Definitions}
In this paper, the Euclidean norm of a vector $x$ and the corresponding induced norm of a matrix $A$ are denoted by $\left\| x \right\|$ and $\left\| A \right\|$, respectively.

\section{Problem Statement} \label{sec: control objective}
%Consider the system with relative degree $n^{*}=1$,
Consider a nonlinear system that, after transformation into normal form, has relative degree one,
\begin{eqnarray}
\dot{x} & = & Ax+B [u+d(t)]\label{eq: x_ponto}\\
y       & = & Cx,\label{eq: y}
\end{eqnarray}
\noindent where $A \in $ $\mathbb R^{n\times n}$, $B \in $ $\mathbb R^n$, $C \in $ $\mathbb R^{1\times n}$ are matrices, $u \in $ $\mathbb R$ is the control input, $d(t)$ is an unmeasured disturbance of the plant, $x \in$ $\mathbb R^n$ is the state vector and $y \in$ $\mathbb R$ is the measured output of system (\ref{eq: x_ponto})--(\ref{eq: y}). %The matrices $A \in $ $\mathbb R^{n\times n}$, $B \in $ $\mathbb R^n$, $C \in $ $\mathbb R^{1\times n}$ are known.

This system can be rewritten in normal form \citep{K:2002} by applying a transformation $T$, described by 
\begin{equation}
\left [
\begin{array}{c}
\eta\\
y\\
\end{array}
\right ] = Tx,
\label{eq:Transformada}
\end{equation}

\noindent with $\eta \in \mathbb R^{n-1}$. Thus, the system (\ref{eq: x_ponto})--(\ref{eq: y}) can be represented by
\begin{eqnarray}
\dot{\eta} & = & A_{0}\eta + B_{0}y \label{eq: eta_ponto}\\
\dot{y} & = & A_{1}\eta + B_{1}y + k_{p}[u+d(t)], \label{eq: y_ponto}
\end{eqnarray}

\noindent where $A_{0} \in $ $\mathbb R^{(n-1) \times (n-1)}$, $A_{1} \in $ $\mathbb R^{1 \times (n-1)}$, $B_{0} \in $ $\mathbb R^{n-1}$ and $B_{1} \in $ $\mathbb R$ are matrices. The parameter $k_{p}$ represents the high frequency gain of the plant, and is also known.
%\footnote{In fact, throughout this work, only the condition of partial knowledge of the parameters of the plant is considered. However, to better present the control strategies developed in this paper, the total knowledge of plant parameters was first considered.}.

The representation in the normal form decomposes the system through an internal dynamics ($\eta$) and an external dynamics ($y$). Let $y = 0$ in
(\ref{eq: eta_ponto}), it follows that

\begin{equation}
\dot{\eta} = A_{0}\eta,
\label{eq: dinamica dos zeros}
\end{equation}

\noindent where this equation is called \textit{zero dynamics} and the eigenvalues of matrix $A_{0}$ are the zeros of the transfer matrix of system (\ref{eq: x_ponto})--(\ref{eq: y}).

Below, follow the Assumptions \textbf{(A1)} and \textbf{(A2)}, relating to the system (\ref{eq: eta_ponto})--(\ref{eq: y_ponto}):

\begin{description}
	\item[\textbf{(A1)}] \textit{On the nonminimum phase}: The matrix $A_{0}$ in (\ref{eq: eta_ponto}) is not Hurwitz.
	\item[\textbf{(A2)}] \textit{On the disturbance}: The disturbance $d(t)$ is uniformly bounded.
\end{description}

%\subsection{Control Objective} \label{sec: control objective}

The objective in this paper is to find an output-feedback control law $u$ so that, from any initial condition, the plant is stabilized, i.e.:
\begin{equation}
T x(t) =
\left [
\begin{array}{c}
\eta (t)\\
y(t)\\
\end{array}
\right ] \to 0,
\label{eq: objetivo}
\end{equation}

\noindent as $t \rightarrow \infty$.

%%%%%%%%%%%%%%%%%%%%%%%%%%%%%%%%%%%%%%%%%%%%%%%%%%%%%%%%%

\section{Control Law Design}\label{sec: projeto de u}

To achieve the objective outlined in the previous section, a sliding variable $s$ is first designed as

\begin{equation}
s := y - K \eta,
\label{eq: s}
\end{equation}

\noindent with $K \in \mathbb R^{1\times n-1}$, following the approach in \citep{GMF:2012}. Its first time derivative is then given by

\begin{equation}
\dot{s} = A_{1} \eta + B_{1} y + k_{p} [u + d(t)] - K A_{0} \eta - K B_{0} y.
\label{eq: s_ponto}
\end{equation}

\noindent If the control law based in sliding modes $u$ is designed as in \citep{UGS:1999}
\begin{equation}
u = - \sgn(k_{p}) \rho \sgn(s),
\label{eq: controle}
\end{equation}

\noindent where $\rho$ is a modulation function defined as
\begin{eqnarray}
\rho &\geq & \frac{1}{k_{min}}  [\left\|A_{1}\right\| \left\| \eta \right\| + \left\|B_{1}\right\| \left| y \right| + \left\|KA_{0}\right\| \left\| \eta \right\| + \nonumber \\
					&	& + \left\|KB_{0}\right\| \left| y \right|] + |d(t)|+ c,
\label{eq: rho}
\end{eqnarray}

\noindent with $k_{min} \in (0,|k_{p}|]$ and $c > 0$ being an arbitrarily small constant, we have that $s \to 0$ in finite time and consequently $y \to K \eta$. So we can rewrite (\ref{eq: eta_ponto}) as
\begin{equation}
\dot{\eta} = A_{0} \eta + B_{0} K \eta = (A_{0} + B_{0} K) \eta.
\label{eq: eta_ponto_1}
\end{equation}

\noindent Thus, if $K$ is designed in such a way that $(A_{0} + B_{0} K)$ is Hurwitz, $\eta \to 0$ exponentially and therefore $y \to 0$.

%%%%%%%%%%%%%%%%%%%%%%%%%%%%%%%%%%%%%%%%%%%%%%%%%%%%%%%%%%%%%%%%%%%%%%%%%%%%%%%%%%%%%%%%%%%%%%%

\section{Observer Design for State and Uncertainty Estimation}\label{sec: mandrake}

The central problem considered herein is the unavailability of the state vector $\eta$ for measurement. Its dynamical evolution is specified by the system given in (\ref{eq: eta_ponto})--(\ref{eq: y_ponto}). The scenario is further complicated by the presence of unknown plant parameters and unmeasured exogenous disturbances. To enable the application of the control law formulated in (\ref{eq: controle}) and (\ref{eq: rho}), the synthesis of observer-based estimators for these uncertainties becomes imperative. Subsequent sections detail the design of these observers, whose estimates are essential for implementing the control strategy and achieving the stabilization objectives of this study.

\subsection{Estimation of Lumped Uncertainties and External Disturbances}\label{obs high gain}

This section outlines the design of two interconnected observers central to our approach. The first observer is tasked with generating an estimate of the output variable $y$. It is pertinent to note that this variable is already measurable. However, this estimated output is expressly utilized as an internal component for constructing a second observer. The purpose of this subsequent observer is to provide simultaneous estimations of the plant's parametric uncertainties and external disturbances.

Consider the following system
\begin{eqnarray}
\dot{\eta} & = & A_{0}\eta + B_{0}y \label{eq: eta_ponto_2}\\
\dot{y} & = & A_{1}\eta + B_{1}y + k_{p} [u + d(t)]. \label{eq: y_ponto_2}
\end{eqnarray}

\noindent Unlike (\ref{eq: y_ponto}), from this point we have $A_{1} = \bar{A_{1}} + \Delta _{A_{1}}$ and $B_{1} = \bar{B_{1}} + \Delta _{B_{1}}$, where the matrices $\bar{A_{1}}$ and $\bar{B_{1}}$ are known and $\Delta _{A_{1}}$ and $\Delta _{B_{1}}$ represent uncertainties around these parameters. Thus, the output (\ref{eq: y_ponto_2}) can be rewritten by

\begin{equation}
\dot{y}  =  \bar{A_{1}}\eta + \bar{B_{1}}y + k_{p} [u + \delta (\eta, y, t)]
\label{eq: novo_y_ponto}
\end{equation}

\noindent where the variable 

\begin{equation}
\delta (\eta, y, t) = d(t) + \frac{\Delta _{A_{1}} \eta + \Delta _{B_{1}} y}{k_{p}}
\label{eq: delta_n}
\end{equation}

\noindent represents the parametric uncertainties of the plant and external disturbances, being actually a lumped uncertainty.

Below, follow the Assumption \textbf{(A3)}, relating to the parametric uncertainties:

\begin{description}
	\item [\textbf{(A3)}]\textit{On the parametric uncertainties}: All of uncertain parameters in (\ref{eq: y_ponto_2}) belongs to a compact set $\Omega$.
\end{description}

For the design of the modulation function $\rho$ in (\ref{eq: rho}), it is necessary to have the measure of the external disturbances $d(t)$, which are not available, obviously. Also, in this section, uncertainties were introduced in the parameters of external dynamics of plant, introducing the term $\delta$ in place of $d$ in (\ref{eq: rho}). Thus, the modulation function $\rho$ can be rewritten as 

\begin{eqnarray}
\rho &\geq & \frac{1}{k_{min}}  [\left\|\bar{A_{1}}\right\| \left\| \eta \right\| + \left\|\bar{B_{1}}\right\| \left| y \right| + \left\|KA_{0}\right\| \left\| \eta \right\| + \nonumber \\
					&	& + \left\|KB_{0}\right\| \left| y \right|] + |\delta(\eta,y,t)|+ c,
\label{eq: rho_mudado}
\end{eqnarray}

\noindent where $\delta (\eta, y, t)$, as well as $d(t)$, must be estimated by an observer.

For estimation of $\delta$, first we design a high-gain observer \citep[Chapter]{K:2002} for the output $y$, described by

\begin{equation}
\dot{\hat{y}} = \bar{A_{1}} \hat{\eta} + \bar{B_{1}} y + k_{p} u - \frac{1}{\varepsilon} (\hat{y} - y)
\label{eq: obs_y}
\end{equation}

\noindent Subsequently, defining a auxiliary variable $l$ as

\begin{equation}
l := \frac{\hat{y}-y}{\varepsilon},
\label{eq: l}
\end{equation}

\noindent it follows that

\begin{eqnarray}
\varepsilon \; \dot{l} & = & \dot{\hat{y}} - \dot{y} \nonumber \\
											 & = & \bar{A_{1}} \hat{\eta} + \bar{B_{1}} y +
															k_{p} u - \frac{1}{\varepsilon} 
															(\hat{y}-y) \nonumber \\
											 &   & 	- \bar{A_{1}} \eta - \bar{B_{1}} y - 
															k_{p} u - k_{p} \delta (\eta , y, t) 															\nonumber \\
											 & = & \bar{A_{1}} \tilde{\eta} - k_{p} \delta 														 (\eta , y, t) - l. \label{cuzao}
\end{eqnarray}

Thus, assuming that the gain $K$, introduced in Section \ref{sec: projeto de u} was appropriately designed, we have that $\tilde{\eta} \to 0$ exponentially and, for sufficiently small $\varepsilon$, so that $\varepsilon \to 0^{+}$, then we can estimate $\delta$ as in \citep{CA:2009}:
\begin{equation}
\lim_{t\to +\infty} \delta (\eta, y, t) =  \frac{-l(t)}{k_{p}} \quad (\text{singular case}~~ \varepsilon \to 0^{+})\,,
\label{eq: delta_singular}
\end{equation}
such that, $\forall \varepsilon \neq 0$:
\begin{equation}
|k_{p}\delta (\eta, y, t)| \leq |l(t)| + \pi_0(t) + \mathcal{O}(\varepsilon)\,,
\label{eq: delta}
\end{equation}
where $\pi_0(t)$ denotes an exponentially decreasing term due to the initial conditions of the filter (\ref{cuzao}).

%An undesired characteristic of the high gain observers is the presence of the phenomenon of peaking \citep{K:2002}, which consists of the presence of a signal similar to a impulse at $t=0$ in the $\delta$ estimate (and consequently in the control input $u$ if this estimate is used in the modulation function $\rho$). However, since the output $y$ can be measured (in fact the output is only observed in order to obtain an estimate for $\delta$), we can define the initial value of $\hat{y}$ the same initial value of $y$, that is, $\hat{y}_{0} = y_{0}$. This initialization eliminates the effects of peaking on the estimation of $\delta$ \citep{SK:1991}.

A well-known limitation of high-gain observers is the occurrence of the peaking phenomenon \citep{K:2002}. This phenomenon manifests as a transient peak, similar to an impulse, at the initial time ($t=0$) in the estimate of $\delta$. Consequently, this peak is also transmitted to the control input $u$ if the estimate is employed within the modulation function $\rho$. To mitigate this issue, one can leverage the fact that the system output $y$ is measurable. By initializing the estimated state variable $\hat{y}$ with the measured value at $t=0$, that is, by enforcing $\hat{y}(0) = y(0)$, the effects of peaking on the estimation of $\delta$ are effectively eliminated \citep{SK:1991}.

%%%%%%%%%%%%%%%%%%%%%%%%%%%%%%%%%%%%%%%%%%%%%%%%%%%%%%%%%%%%

\subsection{Reduced Order Observer for the Unmeasured States of Internal Dynamics}\label{obs ordem red}

We now turn our attention to the design of the observer employed for the reconstruction of the internal state variables, which are not accessible for direct measurement. Consider the following system:
\begin{eqnarray}
\dot{\bar{x}} &=& \bar{A} \bar{x} + \bar{B} \bar{u}\\
\bar{y} &=& \bar{C} \bar{x},
\label{eq: bar_sistema}
\end{eqnarray}

\noindent with $\bar{x} = \eta$, $\bar{A} = A_{0}$, $\bar{B} \bar{u} = B_{0} y$, $\bar{y} = \bar{A_{1}} \eta = \dot{y} - \bar{B_{1}} y - k_{p} [u + \delta (\eta, y, t)]$ and $\bar{C} = \bar{A_{1}}$. According to \citep{C:1999}, we can develop an observer such that 
\begin{eqnarray}
\dot{\hat{\bar{x}}} &=& \bar{A} \hat{\bar{x}} + \bar{B} \bar{u} + L (\bar{y} - \hat{\bar{y}}) \label{eq: observador_puro} \\
\hat{\bar{y}} &=& \bar{C} \hat{\bar{x}}
\end{eqnarray}

\noindent and consequently
\begin{eqnarray}
\dot{\hat{\eta}} &=& A_{0} \hat{\eta} + B_{0} y + L (\bar{y} - \bar{A_{1}} \hat{\eta}) \label{eq: observador_1}\\
\hat{\bar{y}} &=& \bar{A_{1}} \hat{\eta}. \label{eq: observador_2}
\end{eqnarray}

Being the error of observation of the state $\eta$ given by $\tilde{\eta} = \eta - \hat{\eta}$, it follows that
\begin{eqnarray}
\dot{\tilde{\eta}} &=& \dot{\eta} - \dot{\hat{\eta}} \nonumber \\
\dot{\tilde{\eta}} &=& A_{0} \eta + B_{0} y - A_{0} \hat{\eta} - B_{0} y - L (\bar{y} - \bar{A_{1}} \hat{\eta}) \nonumber \\
\dot{\tilde{\eta}} &=& A_{0} (\eta - \hat{\eta}) - L \bar{A_{1}} (\eta - \hat{\eta})\nonumber \\
\dot{\tilde{\eta}} &=& (A_{0} - L \bar{A_{1}}) \tilde{\eta}. \label{eq: contas2}
\end{eqnarray}

Thus, if the gain $L$ is designed so that the matrix $(A_{0} - L \bar{A_{1}})$ is Hurwitz, then 
\begin{equation}
\tilde{\eta} \to 0
\label{eq: eta to 0}
\end{equation}

\noindent and consequently 

\begin{equation}
\hat{\eta} \to \eta.
\label{eq: hateta to eta}
\end{equation}

Despite its proven exponential convergence, the observer presented in (\ref{eq: observador_1})--(\ref{eq: observador_2}) cannot be implemented in its current form. This limitation arises because the signal $\bar{y} = \bar{A}_{1} \eta$ is a function of the internal state $\eta$, which is not available for measurement. To render the observer realizable, we introduce an auxiliary variable $Z$ defined as $Z := \hat{\eta} - L y$, such that 
\begin{eqnarray}
\dot{Z} &=& A_{0} \hat{\eta} + B_{0} y + L \bar{y} - L \bar{A_{1}} \hat{\eta} - L \dot{y} \nonumber \\
\dot{Z} &=& A_{0} \hat{\eta} + B_{0} y + L \dot{y} -            L \bar{B_{1}} y - L \bar{A_{1}} \hat{\eta} \nonumber \\
				& & - L \dot{y} - L k_{p} [u + \delta (\eta, y, t)]             \nonumber \\
\dot{Z} &=& A_{0} \hat{\eta} + B_{0} y - L \bar{B_{1}}             y - L \bar{A_{1}} \hat{\eta} \nonumber \\
				& & - L k_{p} [u + \delta (\eta, y, t)]. \label{eq: dot_z_9}
%\dot{Z}	&=& (A_{0} - L\bar{A_{1}}) \hat{\eta} + (B_{0} -            L\bar{B_{1}})y \nonumber \\
%				& & - L k_{p} [u + \delta (\eta, y, t)]. \nonumber \\
%\dot{Z}	&=& (A_{0} - L\bar{A_{1}}) \hat{\eta} + (B_{0} -            L\bar{B_{1}})y \nonumber \\
%				& & - L k_{p} u + L \, l. \label{eq: z_ponto_2}
\end{eqnarray}

From (\ref{eq: delta_singular}), we substitute the dependence of the term $\delta (\eta, y, t)$ in (\ref{eq: dot_z_9}) by $l$ in (\ref{eq: obs_y})--(\ref{eq: l}) such that:
\begin{eqnarray}
\dot{Z}	&=& (A_{0} - L\bar{A_{1}}) \hat{\eta} + (B_{0} -            L\bar{B_{1}})y \nonumber \\
				& & - L k_{p} u + L \, l. \label{eq: z_ponto_2}
\end{eqnarray}

In this way, a reduced order observer for the unmeasured states of internal dynamics can be implemented as
\begin{equation}
\hat{\eta} = Z + Ly.
\label{eq: observador_eta}
\end{equation}

%%%%%%%%%%%%%%%%%%%%%%%%%%%%%%%%%%%%%%%%%%%%%%%%%%%%%%%%%%%%%%%%%%%%%%%%%%%%%%%%%%%%%%%%%%%%

\subsection{Application of the Estimated Variables in the Control Law}

From the estimators previously presented, we can redesign the equations that define the control law $u$ considering the observed variables. Therefore, (\ref{eq: s}) can be rewritten as

\begin{equation}
s = y - K \hat{\eta}
\label{eq: s_hat}
\end{equation}

\noindent and consequently, considering also that the uncertainties regarding the parameters $\bar{A_{1}}$ and $\bar{B_{1}}$ and the disturbance $d(t)$ are included in term $\delta (\eta, y, t)$, it follows that

\begin{equation}
\dot{s} = \bar{A_{1}} \eta + \bar{B_{1}} y + k_{p} [u + \delta ] - K A_{0} \hat{\eta} - K B_{0} y.
\label{eq: s_ponto_hat}
\end{equation}

Thus, the modulation function $\rho$ in (\ref{eq: rho}) must be redefined as
\begin{eqnarray}
\rho &\geq & \frac{1}{k_{min}}  [\left\|\bar{A_{1}}\right\| \left\| \hat{\eta} \right\| + \left\|\bar{B_{1}}\right\| \left| y \right| + \left\|KA_{0}\right\| \left\| \hat{\eta} \right\| \nonumber \\
					&	& + \left\|KB_{0}\right\| \left| y \right| + |l|] + c,
\label{eq: rho_2}
\end{eqnarray}

\noindent with $k_{min} \in (0,|k_{p}|]$ and $c > 0$ a design constant.

%%%%%%%%%%%%%%%%%%%%%%%%%%%%%%%%%%%%%%%%%%%%%%%%%%%%%%%%%%%%%%%%%%%%%%%%%%%%%%%%%%%%%%%%%%%%%%%%%%%%%%%%%%%%%%%%%%%%%%%%%%%%%%%

\subsection{Stability Analysis}
%The following theorem presents a possible implementation of the modulation function such that (\ref{eq: rho_2}) is checked and the sliding mode controller through multiple observers and output feedback proposed stabilize the output $y$ and the states $\eta$ in finite time.

The subsequent theorem delineates a feasible design for the modulation function. This design ensures that the condition specified in (\ref{eq: rho_2}) is satisfied and, consequently, guarantees the finite-time stabilization of both the output $y$ and the state variables $\eta$ under the proposed multi-observer-based sliding mode control utilizing output feedback.

%\begin{theorem}
%Consider the normal form system described by (\ref{eq: eta_ponto_2})--(\ref{eq: y_ponto_2}), which is governed by the control law (\ref{eq: controle}) and the dynamics of the auxiliary variable $s$ given in (\ref{eq: s_ponto_hat}). The adaptive observer architecture consists of a reduced-order observer defined by (\ref{eq: z_ponto_2}) and (\ref{eq: observador_eta}), a high-gain observer (\ref{eq: obs_y}), and an estimator for parametric uncertainties and disturbances specified in (\ref{eq: l}) and (\ref{eq: delta}). Let $\varepsilon>0$ be chosen sufficiently small. Under the validity of assumptions \textbf{(A1)}, \textbf{(A2)}, and \textbf{(A3)}, and provided the modulation function $\rho$ in (\ref{eq: controle}) adheres to condition (\ref{eq: rho_2}) and is defined as
%\begin{eqnarray}
%\rho & = & \frac{1}{|k_{p}|}  [\left\|\bar{A_{1}}\right\| \left\| \hat{\eta} \right\| + \left\|\bar{B_{1}}\right\| \left| y \right| + \left\|KA_{0}\right\| \left\| \hat{\eta} \right\| \nonumber \\
%					&	& + \left\|KB_{0}\right\| \left| y \right| + |l|] + c,
%\label{eq: rho_3}
%\end{eqnarray}
%
%\noindent where $c > 0$ is an arbitrary small constant, the system is globally asymptotically stabilized. This result holds if the gain matrices $K$ (\ref{eq: eta_ponto_1}) and $L$ (\ref{eq: observador_eta}) are selected to ensure that $(A_{0} + B_{0} K)$ and $(A_{0}-L\bar{A_{1}})$ are Hurwitz, respectively. Consequently, the system output $y$ and the state vector $\eta$ converge to zero globally and asymptotically.
%\end{theorem}

\begin{theorem}
Consider the normal-form system (\ref{eq: eta_ponto_2})--(\ref{eq: y_ponto_2}) under Assumptions (A1)--(A3). Let the control law be given by (\ref{eq: controle}), where the modulation function is chosen as

\begin{eqnarray}
\rho & = & \frac{1}{|k_{p}|}  [\left\|\bar{A_{1}}\right\| \left\| \hat{\eta} \right\| + \left\|\bar{B_{1}}\right\| \left| y \right| + \left\|KA_{0}\right\| \left\| \hat{\eta} \right\| \nonumber \\
					&	& + \left\|KB_{0}\right\| \left| y \right| + |l|] + c,
\label{eq: rho_3}
\end{eqnarray}

with $c>0$. Suppose that

\begin{enumerate}
\item the gain matrix $K$ is selected such that $(A_0+B_0K)$ is Hurwitz;

\item the observer gain $L$ is selected such that $(A_0-L\bar A_1)$ is Hurwitz;

\item the observer consists of the high-gain observer (\ref{eq: obs_y}), the uncertainty estimator (\ref{eq: l})-(\ref{eq: delta}), and the reduced-order observer (\ref{eq: z_ponto_2})-(\ref{eq: observador_eta});

\item the design parameter $\varepsilon>0$ is chosen sufficiently small.
\end{enumerate}

Then, the sliding variable satisfies
\[
s(t)\rightarrow0,
\]
and the equilibrium
\[
(\eta,y)=(0,0)
\]
is globally asymptotically stable. Consequently,
\[
\eta(t)\rightarrow0,
\qquad
y(t)\rightarrow0,
\]
as $t\rightarrow\infty$.
\end{theorem}

\begin{proof} The next analysis follows the standard singular perturbation argument \citep[Chapter~5]{utkin1992sliding}, \citep{sg2} commonly adopted for high-gain observer-based designs, such as  \citep[Chapter~14]{K:2002} and \citep{sg3}, where the observer parameter $\varepsilon$ is selected sufficiently small to render the approximation error arbitrarily small. Let the following Lyapunov candidate function
\begin{equation}
V = \frac{s^{2}}{2}.
\label{eq: lyap}
\end{equation}

Differentiating $V$ along the closed-loop trajectories yields
\begin{eqnarray}
\dot{V} & = & s \dot{s} \nonumber \\
				& = & s(\dot{y}-K\dot{\hat{\eta}}) \nonumber \\
				& = & s(\bar{A_{1}} \eta \!+\! \bar{B_{1}} y \!+\! k_{p}(u+\delta) \!-\! KA_{0} \hat{\eta} \!-\! KB_{0}y \nonumber \\
				&		&	- KL( \underbrace{\bar{y}}_{=\bar{A_{1}} \eta} - \bar{A_{1}} \hat{\eta})) \label{eq: lyap_2}
%        & = & s(\bar{A_{1}} \eta + \bar{B_{1}} y - KA_{0} \hat{\eta} - KB_{0}y + k_{p} \delta
%				&   & \!+\! \,  k_{p}(\!-\!\rho sgn(k_{p}) sgn(s)) \!-\!
%				{KL(\bar{A_{1}} \eta \!-\! \bar{A_{1}} \hat{\eta})}).
\end{eqnarray}

substituting (\ref{eq: controle}) into (\ref{eq: lyap_2}), we obtain:

\begin{eqnarray}
\dot{V} & = & s(\bar{A_{1}} \eta + \bar{B_{1}} y - KA_{0} \hat{\eta} - KB_{0}y + k_{p} \delta + \nonumber \\
				&   & k_{p}(\!-\!\rho \sgn(k_{p}) \sgn(s)) \!-\!
				{KL(\bar{A_{1}} \eta \!-\! \bar{A_{1}} \hat{\eta})}).
\label{eq: lyap_3}
\end{eqnarray}

From (\ref{eq: contas2}) and (\ref{eq: eta to 0}), we have that term $KL(\bar{A_{1}} \eta - \bar{A_{1}} \hat{\eta}) = KL\bar{A_{1}} \, \tilde{\eta} \to 0$ exponentially. Thus, there exists a bounded exponentially decaying function $\pi _{1} := -KL(\bar{A_{1}} \eta - \bar{A_{1}} \hat{\eta})$ as an exponential decay term.  So, using the modulation function $\rho$ given in (\ref{eq: rho_3}), it follows that

\begin{eqnarray}
%				&   & + \,  k_{p}(-\rho sgn(k_{p}) sgn(s)) - l -
%				\underbrace{KL(\bar{A_{1}} \eta - \bar{A_{1}} \hat{\eta})}_{(KL\bar{A_{1}} \tilde\eta)\to 0}) \\
%				& = & s(\bar{A_{1}} \eta + \bar{B_{1}} y - KA_{0} \hat{\eta} - KB_{0}y +\\
%				&   & +  |k_{p}| (-\rho sgn(s)) -l) \\
\dot{V}	& = & s \, \bar{A_{1}} \eta \!-\! \left| s \right| \left\| \bar{A_{1}} \right\| \left\| \hat{\eta} \right\| 
				\!+\! s \,  \bar{B_{1}} y \!-\! \, \left| s \right| \left\| \bar{B_{1}} \right\| \left| y \right| - \nonumber \\
				&   &  s \, KA_{0} \hat{\eta} - \left| s \right| \left\| KA_{0} \right\| \left\| \hat{\eta} \right\|
				- s \, KB_{0} y - \nonumber \\
				&   & \left| s \right| \! \left\| KB_{0} \right\| \! \left| y \right| \!+\!  s k_p \delta 
				\!-\! \left| s \right| \! \left| l \right| \!-\! \left| s \right| \! \left| k_{p} \right| c \!+\! s \pi _{1}.
\end{eqnarray}

Applying the triangle inequality together with the induced matrix norm, it follows that

\begin{eqnarray}
\dot{V}	& \leq & \left| s \right| \left\| \bar{A_{1}} \right\| \left\| \eta \right\| - \left| s \right| \left\| \bar{A_{1}} \right\| 
								\left\| \hat{\eta} \right\| + \nonumber \\
				&		&  \left| s \right| \left\| \bar{B_{1}} \right\| \left| y \right|	- \, \left| s \right| \left\| \bar{B_{1}} \right\| 
				        \left| y \right| + \nonumber \\
				&   &  \left| s \right| \left\| KA_{0} \right\| \left\| \hat{\eta} \right\| 
							 - \left| s \right| \left\| KA_{0} \right\| \left\| \hat{\eta} \right\| + \nonumber \\
				&	  &  \left| s \right| \left\| KB_{0} \right\| \left| y \right| -\left| s \right| \left\| KB_{0} \right\| \left| y \right| + \nonumber \\
				&   &  \left| s \right| \left| k_p \delta \right| - \left| s \right| \left| l \right| 
						  - \left| s \right| \left| k_{p} \right| c + s \pi _{1}.
\end{eqnarray}

Still, from (\ref{eq: delta}), we can establish an upper bound to the term $|k_{p} \delta|-|l| \leq \pi_0+\mathcal{O}({\varepsilon})$, being $\pi_0$ an exponential decay term. Therefore, it follows that

\begin{eqnarray}
\dot{V}	& \leq & \left| s \right| \left\| \bar{A_{1}} \right\| (\left\| \eta \right\| - \left\| \hat{\eta} \right\|) \nonumber \\
		  &   & + \left| s \right| \left[ \pi_0 + \mathcal{O}(\varepsilon) \right] 
%							- 2 \left| s \right| \left\| KB_{0} \right\| \left| y \right| \\
%				&   & - 2 \left| s \right| \, \left| l \right|
- \left| s \right| \left| k_{p} \right| c + s \pi _{1}.	
\end{eqnarray}

Furthermore, from (\ref{eq: contas2}) and (\ref{eq: eta to 0}), we have
\begin{eqnarray}
\left\|\bar{A_{1}} \right\| (\left\| \eta \right\| - \left\| \hat{\eta} \right\|) &=& \left\| \bar{A_{1}} \right\| (\left\| \tilde{\eta} + \hat{\eta} \right\|-\left\| \hat{\eta} \right\|) \\
& \leq & \left\| \bar{A_{1}} \right\| \left\| \tilde{\eta} \right\| \to 0 \text{ (exponentially)}. 
\end{eqnarray}
Thus, by defining $\pi _{2} := \left| s \right| \left\| \bar{A_{1}} \right\| (\left\| \eta \right\| - \left\| \hat{\eta} \right\|)$ also as a exponential decay term, we obtain
\begin{eqnarray}
\dot{V}	& \leq & \left| s \right| \left[ \pi _{2} + \pi_0 + \mathcal{O}(\varepsilon)  - \left| k_{p} \right| c + \sgn(s) \pi _{1} \right].	
\end{eqnarray}
Since $\pi _{0}$, $\pi _{1}$ and $\pi _{2}$ are terms of exponential decay, we can say that $\exists \, t  \leq t_{1}$, where $t_{1} > 0$ is an instant of finite time, such that $\pi _{2} + \pi_0 + \mathcal{O}(\varepsilon) + \sgn(s) \pi _{1} < \left| k_{p} \right| c $, for $\varepsilon>0$ sufficiently small. 
%
%\begin{equation}
%\pi _{2} + s \pi _{1} < \left| s \right| \left| k_{p} \right| c + 2 \left| s \right| \left\| KA_{0} \right\| \left\| \hat{\eta} \right\| +
%2 \left| s \right| \left\| KB_{0} \right\| \left| y \right| + 2 \left| s \right| \, \left| l \right|.
%\label{eq: contas_69}
%\end{equation}
%
\noindent Since $c > 0$, we conclude that
\begin{equation}
\dot{V} < 0,
\label{eq: prova}
\end{equation}
\noindent and $s \to 0$ asymptotically. From (\ref{eq: contas2}), (\ref{eq: eta to 0}) and (\ref{eq: s_hat}), consequently $y \to K \eta$. Therefore, from (\ref{eq: eta_ponto_1}), it follows that $\eta \to 0$. Thus, the equilibrium point $\eta = 0$ and $y = 0$ is global asymptotically stable. 

\end{proof}

\section{Simulation Results}

%\begin{eqnarray}
%\left [ 
%\begin{array}{c}
%	\dot{\eta}\\
%	\dot{y}\\
%\end{array} \right ] &=&
%\left [ 
%\begin{array}{rrr}
%	0 & 1 & 0 \\
% -1 & -1 & 1  \\
%	0 & 1 & 1 \\
%\end{array} 
%\right ]  \left [ 
%\begin{array}{c}
%	\eta\\
%	y \\
%\end{array} \right ]  +
%\left [ 
%\begin{array}{c}
%	0 \\
%	0 \\
%	0,5 \\
%\end{array} 
%\right ] & u \nonumber \\
%y&=& \left[ \begin{array}{ccccc}
%	0 & \, \, \, \, \, \,  0  & \, \, \, \, \, \, 1 \\
%\end{array} \right]  \left [ 
%\begin{array}{c}
%	\eta\\
%	y \\
%\end{array} \right ] . \nonumber
%\end{eqnarray}

This section details the simulation outcomes for the system described by equations (\ref{eq: eta_ponto_2})--(\ref{eq: y_ponto_2}), which is considered in its normal form. The applied control strategy integrates the control law from (\ref{eq: controle}) with the modulation function specified in (\ref{eq: rho_3}) and employs the observers developed in Sections \ref{obs high gain} and \ref{obs ordem red}.

The primary control aim, as established in Section \ref{sec: control objective}, is the global finite-time stabilization of the system states. This implies that the control signal $u$ must drive both $\eta$ and $y$ to converge to zero within a finite time horizon, starting from any initial condition.

The simulation parameters were selected as follows: $A_{0} = 1$, $B_{0} = -1$, $\bar{A_{1}} = 1$, $\bar{B_{1}} = 1$, and $k_{p} = 1$. To ensure the Hurwitz stability of the matrices $(A_{0} + B_{0} K)$ and $(A_{0} - L \bar{A_{1}})$, the gains were set to $K = 3$ and $L = 3$. Parameter uncertainties in $A_{1}$ and $B_{1}$ were modeled by defining $\Delta {A{1}} = 1$ and $\Delta {B{1}} = 0.5$, respectively. The disturbance signal $d(t)$ affecting (\ref{eq: y_ponto_2}) was chosen as a sinusoid with an amplitude of $2$ and a frequency of $10$ rad/s. The plant's initial conditions were set to $\eta_ {0} = 3$ and $y_{0} = 3$.

The system's output response is shown in Fig. \ref{fig: saida da planta}. It can be seen that, following a brief transient period, the control input $u$ successfully stabilizes the output $y$ at the origin. This confirms that the central objective of the proposed control algorithm has been achieved for the system (\ref{eq: eta_ponto_2})--(\ref{eq: y_ponto_2}).

The same figure also displays the performance of the high-gain observer formulated in (\ref{eq: obs_y}). This observer provides a satisfactory estimate of the plant's measurable output $y$. As anticipated and discussed in Section \ref{sec: mandrake}, the main purpose of this observer is not to estimate the measurable output itself, but to serve as a foundational tool for constructing an estimator of the combined parametric uncertainties and external disturbances.

Furthermore, initializing the observer state $\hat{y}$ with the same value as $y(0)$ proved effective in mitigating the peaking phenomenon. This is corroborated by the absence of extreme initial peaks in the control input signal $u$, as illustrated in Fig. \ref{fig: controle}.

A comparison between the internal state variable $\eta$ and its estimate $\hat{\eta}$ is provided in Fig. \ref{fig: eta}, with the corresponding estimation error $\tilde{\eta}$ shown in Fig.~\ref{fig: etatil}. It is noteworthy that the observer was initialized with a value different from the actual state (which is unmeasurable), yet it demonstrated satisfactory convergence and estimation performance. The control input $u$ also effectively stabilizes the state $\eta$ after the initial transient.

Finally, the estimation of the aggregated uncertainty $\delta$ is compared to the disturbance input $d$ in Fig. \ref{fig: delta}. After the transient phase, the estimate $\hat{\delta}$ closely tracks the disturbance signal $d$. This behavior is consistent with the theoretical formulation in (\ref{eq: delta_n}): since the parametric uncertainty term $(\Delta {A{1}} \eta + \Delta {B{1}} y)/k_{p}$ vanishes as $\eta$ and $y$ converge to zero, the total uncertainty $\delta$ asymptotically approaches $d$.

It is worth observing that stabilization is achieved despite the simultaneous presence of plant uncertainties and external disturbances, indicating that the observer-controller interaction remains effective throughout the transient response.

\begin{figure}
\centering
\subfloat[Output $y$ and its estimate $\hat{y}$]{%
\label{fig: saida da planta}
\resizebox*{6.5cm}{!}{\includegraphics[width=\columnwidth]{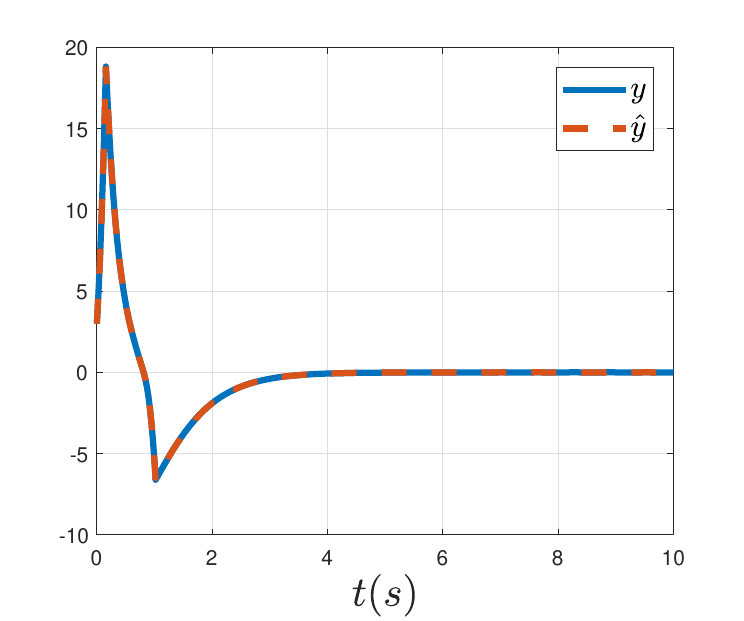}}}\hspace{5pt}
\subfloat[Control input $u$.]{%
\label{fig: controle}
\resizebox*{7.1cm}{!}{\includegraphics{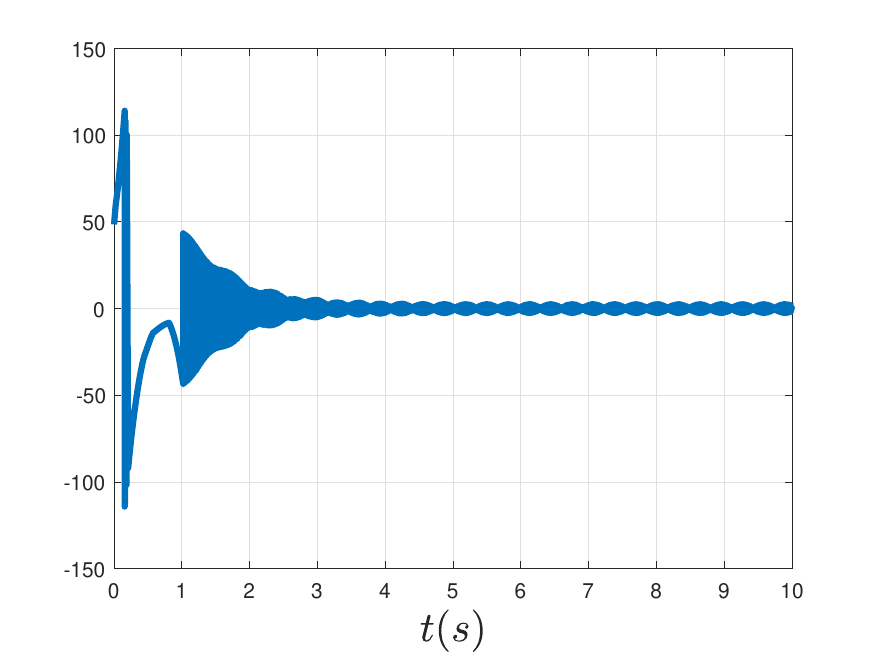}}}
\caption{(a) System output response $y$ and its estimate $\hat{y}$ from the high-gain observer. After a brief transient, the control input successfully drives the measurable output $y$ to zero, achieving the primary control objective. (b) The signal $u$ is free of extreme initial peaks, demonstrating that initializing the observer $\hat{y}$ with the same value as $y(0)$ effectively mitigated the peaking phenomenon.} 
\end{figure}

\begin{figure}
\centering
\subfloat[State $\eta$ and its estimate $\hat{\eta}$.]{%
\label{fig: eta}
\resizebox*{6.5cm}{!}{\includegraphics[width=\columnwidth]{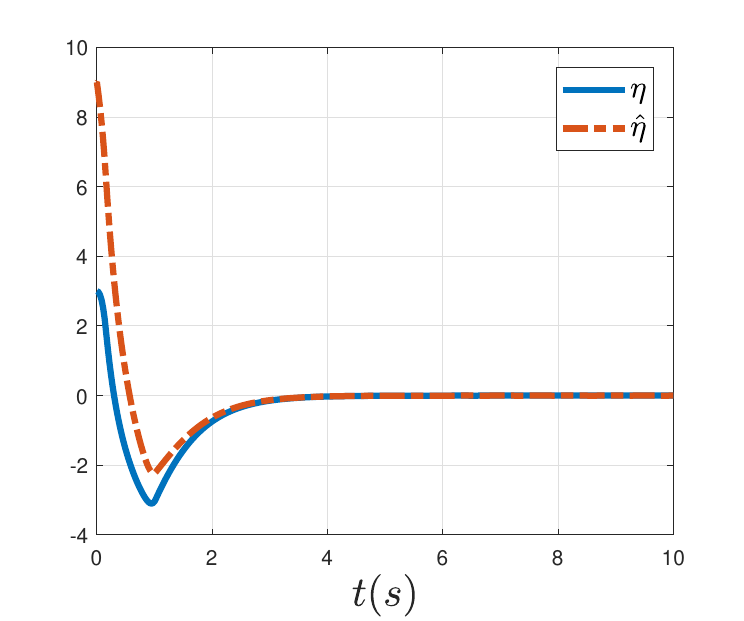}}}\hspace{5pt}
\subfloat[Estimation error $\tilde{\eta}$.]{%
\label{fig: etatil}
\resizebox*{7.1cm}{!}{\includegraphics{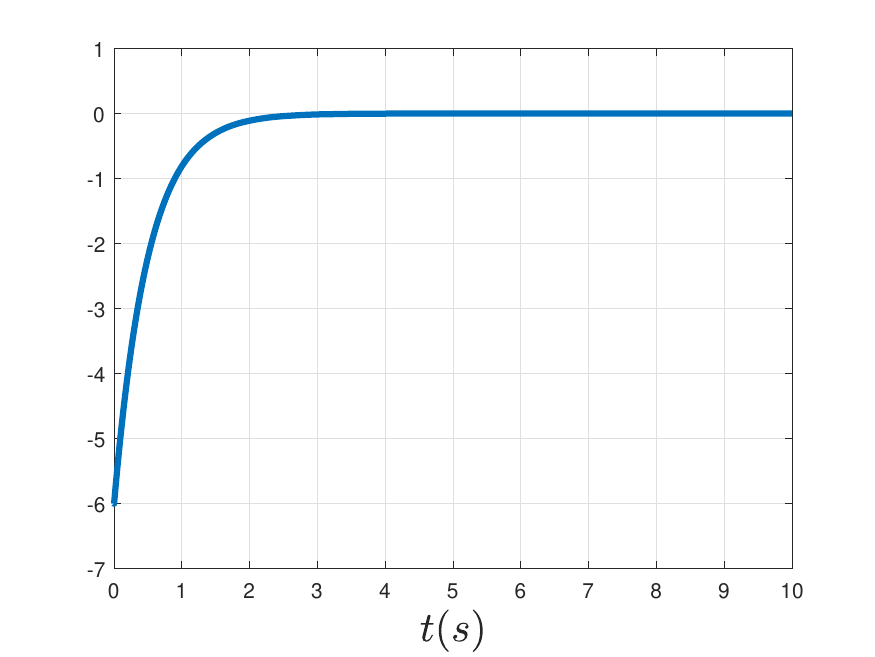}}}
\caption{(a) Internal state $\eta$ and its estimate $\hat{\eta}$ from the reduced-order observer. Despite being initialized with a value different from the actual (unmeasurable) state, the observer demonstrates satisfactory convergence and performance. (b) Estimation error $\tilde{\eta} = \eta - \hat{\eta}$ of the internal state. The error converges satisfactorily, confirming the effective performance of the state observer.} 
\end{figure}

\begin{figure}[ht]
\centering
%\vspace{-1.5cm}
%\hspace{-1cm}
\includegraphics[scale=0.5]{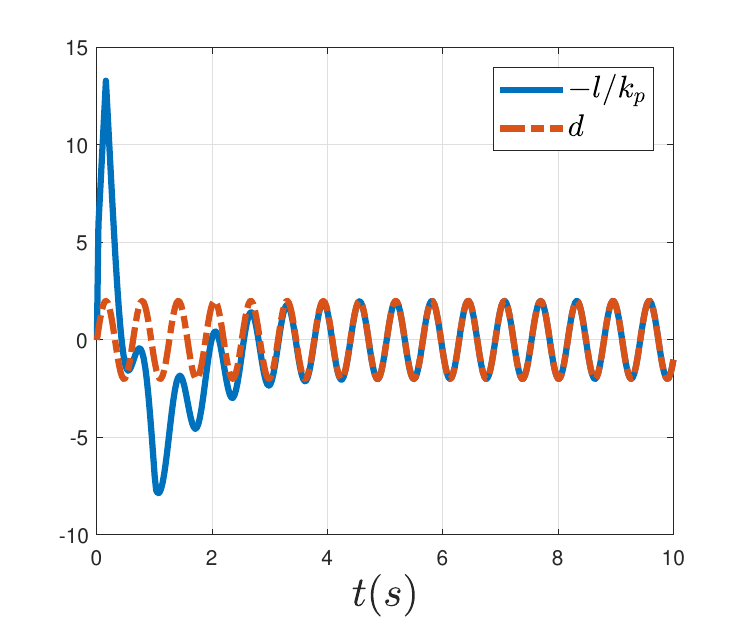}
\caption{Aggregated uncertainty $\delta$ (comprising parametric uncertainties and external disturbance $d(t)$) and its estimate $\hat{\delta}$. After the transient phase, $\hat{\delta}$ closely tracks the disturbance signal $d$, as the parametric uncertainty term vanishes with the convergence of $\eta$ and $y$ to zero.}
\label{fig: delta}
\end{figure}

\section{Conclusion}
This paper has presented a novel output-feedback control strategy for the global stabilization of uncertain nonminimum phase systems. The core contribution is a multi-observer framework that effectively circumvents the fundamental obstacle posed by unstable zero dynamics and the unavailability of the full state vector for measurement. The proposed solution decomposes the problem into manageable estimation tasks: a reduced-order observer reliably reconstructs the internal state, while a dedicated estimation mechanism accurately identifies the aggregated parametric and disturbance uncertainties. By feeding these estimates into a sliding mode control law with a carefully designed modulation function, global asymptotic stability is guaranteed.

Numerical simulations on an academic example validated the theoretical developments. The results demonstrated excellent performance: the output and internal states were successfully driven to zero; the observers exhibited rapid and accurate convergence, even when initialized with errors; the peaking phenomenon was effectively mitigated; and the estimate of the aggregated uncertainty converged to the true disturbance signal as the states approached the origin.

The proposed architecture is intentionally modular, allowing different observer designs to be incorporated without modifying the overall controller structure. Consequently, improvements in disturbance estimation or output differentiation techniques may be readily integrated into the framework while preserving the basic control philosophy.

This work opens promising avenues for future research. Immediate extensions include generalizing the approach for systems with higher relative degrees and extending the stability analysis to provide formal guarantees for the transient performance and ultimate bounds in the presence of the singular perturbation argument. Furthermore, experimental validation on physical NMP systems would be a valuable step towards practical application. The present work focuses on systems of relative degree one and assumes the availability of the normal-form representation. Extension to higher relative-degree systems constitutes an important direction for future work.

%\section*{Acknowledgments}
%
%This work is dedicated to the memory of Professor Shankar Prashad Bhattacharyya from Texas A\&M University, who visited UFRN several times during an American student exchange program. We are deeply grateful for his generosity in sharing his expertise through insightful lectures and courses, as well as for the invaluable technical discussions and warm, informal conversations that fostered a true sense of international camaraderie. Beyond his academic contributions, we were uniquely enriched by his artistic talent, as he generously shared the beautiful sounds of Indian classical music through his captivating performances on the sarod. Professor Bhattacharyya was a remarkable bridge between cultures, and he will be remembered with great respect and affection.
%

\section*{Declarations}
	
\noindent\textbf{Conflict of interest/Competing interests:}	The authors declare that they have no known competing financial interests or personal relationships that could have appeared to influence the work reported in this paper.

\noindent\textbf{Funding:}	This study was financed in part by the Coordenação de Aperfeiçoamento de Pessoal de Nível Superior - Brasil (CAPES) - Finance Code 001.
	
\noindent\textbf{Author Contributions:} Author contributions are as follows: R.S. conceived the study, developed the methodology, conducted the formal analysis and investigation, and wrote the original draft. K.Q., S.D. and T.R.O., contributed to the methodology, reviewed and edited the manuscript. R.S. developed the software, created the visualizations, and reviewed and edited the manuscript. All authors discussed the results, contributed to the final manuscript, and approved the final version.
\bibliographystyle{spbasic}      
\bibliography{myinteractapasample}

\end{document}